\documentclass{ieeeaccess}
\usepackage{cite}
\usepackage{amsmath,amssymb,amsfonts}
\usepackage{algorithm}

\usepackage{algcompatible}
\usepackage{listings}
\usepackage{textcomp}
\usepackage{longtable}
\usepackage{tikz}
\usepackage{tikz-cd}
\usepackage{graphicx}

\usepackage{pgfplots}
\usetikzlibrary{shapes,arrows}
\usetikzlibrary{pgfplots.groupplots}
\usepackage{colortbl}
\usepackage{arydshln}
\NewSpotColorSpace{PANTONE}
\AddSpotColor{PANTONE} {PANTONE3015C} {PANTONE\SpotSpace 3015\SpotSpace C} {1 0.3 0 0.2}
\SetPageColorSpace{PANTONE}%

\usepackage{multirow}
\usepackage{tabularx,booktabs}
\newcolumntype{L}[1]{>{\raggedright\arraybackslash}p{#1}}
\newcolumntype{C}[1]{>{\centering\arraybackslash}p{#1}}
\newcolumntype{R}[1]{>{\raggedleft\arraybackslash}p{#1}}
\usepackage{url}
\usepackage{soul}
\usepackage{longtable}
\usepackage{array}
\usepackage{arydshln}
\usepackage{amsthm}

\newtheorem{theorem}{Theorem}
\newtheorem{remark}{Remark}

\newtheorem{problem}{Problem}

\usepackage{amsmath}
\usepackage{amssymb}
\usepackage{amsfonts}
\usepackage{amscd}
\usepackage{array}
\usepackage{multirow}
\usepackage{booktabs}
\usepackage{caption}
\usepackage{graphicx}

\makeatletter

\newcommand{\thickhline}{%
    \noalign {\ifnum 0=`}\fi \hrule height 0.6pt
    \futurelet \reserved@a \@xhline
}
\newcolumntype{"}{@{\hskip\tabcolsep\vrule width 1pt\hskip\tabcolsep}}

\makeatother

\def\BibTeX{{\rm B\kern-.05em{\sc i\kern-.025em b}\kern-.08em
    T\kern-.1667em\lower.7ex\hbox{E}\kern-.125emX}}
    
\begin{document}
\history{Date of publication xxxx 00, 0000, date of current version xxxx 00, 0000.}
\doi{10.1109/ACCESS.2017.DOI}

\title{POKEx: Performance analysis of POKE-key exchange and SIDH-variants }
\author{\uppercase{Hyeonhak Kim\authorrefmark{1}}, \uppercase{Suhri Kim\authorrefmark{2}}}
\address[1]{School of Cybersecurity, Korea University, Seoul 02841, South Korea}
\address[2]{School of Mathematics, Statistics and Data Science, Sungshin Women’s University, Seoul, 02844, South Korea}
\tfootnote{This work was supported by the Sungshin Women's University Research Grant of 2024 (H20240079).}

\markboth
{Author \headeretal: Preparation of Papers for IEEE TRANSACTIONS and JOURNALS}
{Author \headeretal: Preparation of Papers for IEEE TRANSACTIONS and JOURNALS}

\corresp{Corresponding author: Suhri Kim (e-mail: suhrikim@sungshin.ac.kr).}

\begin{abstract}
In this paper, we present a comparative performance analysis of the POK\'E-based key exchange and SIDH variants. SIDH gained attention for its small key size and efficient performance, and has been selected as an alternate candidate in NIST PQC Round 4. However, following the key recovery attack by Castryck and Decru in 2022, SIDH was shown to be vulnerable to polynomial-time attacks on classical computers, undermining its security and removing it from consideration. Given that SIDH was regarded as the leading isogeny-based algorithm, several countermeasures, such as MSIDH, MD-SIDH, and bin/terSIDH have been proposed. However, these approaches still face performance limitations. Meanwhile, POK\'E, proposed by Basso and Maino at Eurocrypt 2025, is an isogeny-based public key encryption scheme that combines a SIDH-like protocol with higher-dimensional isogenies and has drawn attention for its efficient performance. In this work, we adapt POK\'E into a key exchange algorithm and benchmark it against M-SIDH, terSIDH, and CSIDH. Targeting NIST security level 1, POK\'E-based KEM is approximately 21.21 times faster than terSIDH and 64.97 times faster than CSIDH, showing that the POK\'E-based KEM is currently the most promising isogeny-based key exchange candidate.
\end{abstract}


\titlepgskip=-15pt

\maketitle

\section{Introduction}
\label{intro}
The start of isogeny-based cryptography can be traced back to the work of Couveignes in \cite{couveignes2006hard}, which is now known as the CRS scheme. However, it was the introduction of SIDH (Supersingular Isogeny Diffie-Hellman) by De Feo and Jao that marked the beginning of active research in this area \cite{jao2011towards}. In contrast to the original CRS scheme and standard elliptic curve cryptography, which are based on ordinary elliptic curves, SIDH uses supersingular elliptic curves. As the endomorphism ring of the supersingular elliptic curve is non-commutative, this allows SIDH to withstand the attack proposed in \cite{childs2014constructing} that solves the CRS scheme with subexponential complexity. Moreover, various optimization efforts, including the influential work by Costello et al. in \cite{costello2016efficient}, lead to efficient performance, which significantly contributed to the proposal of the encapsulation algorithm SIKE in the NIST PQC standardization project. Unlike many other algorithms submitted to the PQC standardization project, which required increasingly larger key or signature sizes due to newly discovered attacks, SIKE stood out by maintaining resistance against all known attacks for nearly a decade. In fact, since its introduction, no attack more efficient than the Meet-in-the-Middle (MitM) approach was found \cite{costello2021case}, and analysis showed that smaller parameters could be used while maintaining the same security level, which strengthens the SIKE's advantage in terms of key size. Some studies even suggested that, for SIDH-based algorithms, classical attacks could be more effective than quantum ones, further supporting its consideration as an alternate candidate in Round 4.

However, in 2023, Castryck and Decru showed that the private key in SIDH can be recovered in classical polynomial time \cite{castryck2023efficient}. Due to the non-commutative property of supersingular elliptic curves, SIDH requires publishing the computation of other parties' public keys with one's private key. This is called \textit{torsion point information}. Torsion-point information is not seen in general key exchange algorithms and is unique to SIDH-based cryptography. Hence, a significant amount of research has targeted this feature as a potential weakness \cite{petit2017faster,de2021improved,de2021seta,fouotsa2022isogeny}. The Castryck–Decru attack in 2023 was the first to fully exploit it, demonstrating that SIDH can be broken in polynomial time \cite{castryck2023efficient}. The attack uses the torsion point information together with Kani's glue-and-split algorithm. Because this attack fundamentally exploits torsion point information, cryptosystems such as SIDH, BSIDH, and SETA that use torsion point information can be broken in polynomial time. Subsequently, Maino et al. \cite{maino2023direct} and Robert \cite{robert2023breaking} independently showed that SIDH-based algorithms remain vulnerable even with arbitrary starting curves, confirming that SIDH-based schemes are no longer secure.

As SIDH is regarded as one of the most efficient among isogeny-based cryptosystems, several studies have been conducted to preserve its advantages while mitigating the Castryck–Decru attack. The Castryck-Decru attack relies on three main facts: 1) torsion-point information that SIDH-based cryptography publishes, 2) the degree of the isogeny connecting two elliptic curves is known, and 3) this degree is power-smooth for efficiency. Among these, torsion-point information plays the most critical role. To address these points, \cite{fouotsa2023m} proposed two countermeasures: M-SIDH (masked torsion point SIDH), which hides torsion-point information by masking it with random values, and MD-SIDH (masked-degree SIDH), which masks the degree of the isogeny used. However, unlike conventional SIDH-based cryptosystems, which use primes of the form $p=2^{e_1}3^{e_2}-1$ that enable efficient reduction, M-SIDH and MD-SIDH adopt primes of the form $p=\prod \ell_i-1$, where $\ell_i$s are odd primes. This choice prevents efficient reduction and requires primes more than ten times larger for the same security level (M-SIDH requires a 5911-bit prime for NIST security level 1), resulting in reduced efficiency.

In \cite{basso2023new}, Basso et al. proposed binSIDH (binary SIDH) and terSIDH (ternary SIDH) to counter the key recovery attack. Both binSIDH and terSIDH introduce the concept of artificial orientation and adopt new forms of secret keys. Among them, terSIDH provides about 1.6 bits of security per private key bit and offers relatively fast performance, making it a competitive SIDH variant. However, the total isogeny degree of terSIDH depends on the private key, causing variations in execution time and exposing it to potential timing attacks. Overall, while several SIDH variants have been proposed to address the Castryck–Decru attack, none have proven to be truly practical, highlighting the need for exploring a wider range of scheme designs.

Meanwhile, the Castryck–Decru attack prompted a major redesign of SIDH-based cryptosystems. Traditionally, isogenies were represented by their corresponding cyclic kernel. Knowing the kernel allows the isogeny to be computed using Vélu’s formulas, and the smooth order of torsion points is required for an efficient isogeny computation. However, key-recovery attacks have shown that it is possible to perform an efficient isogeny computation even when the degree is not smooth.
Hence, knowing a pair of curves, their associated torsion points, and the isogeny degree is sufficient to compute the connecting isogeny. Applying this new high-dimensional approach has yielded more efficient results, as seen in schemes such as SQIsign and SCALLOP \cite{basso2024sqisign2d,nakagawa2024sqisign2d,duparc2024sqiprime,chen2024scallop,panny2024klapoti}. However, applying such a high-dimensional approach to SIDH to counter the key recovery attack was first done in \cite{basso2025poke}.

In \cite{basso2025poke}, Basso and Maino proposed POK\'E, an isogeny-based encryption algorithm that constructs a SIDH-like commutative diagram integrating a two-dimensional representation. As mentioned above, SIDH variants designed to counter existing key-recovery attacks are not only defined over larger finite fields than SIDH, but also use field structures that are not reduction-efficient, which has the greatest impact on their performance. On the other hand, POK\'E is based on a finite field with a prime of the form $p=2^a 3^b 5^c f-1$, a structure that is central to its high efficiency. In POK\'E, one party computes isogenies of non-smooth degree using the two-dimensional representation, while the other computes smooth-degree isogenies in the traditional way. For a security parameter $\lambda$, POK\'E requires a prime of size $10\lambda/3$ bits. This means that current public keys and ciphertexts are small and compare favorably with most protocols in the literature, resulting in very efficient performance.

\subsection{Our Contributions}

 As noted above, POK\'E benefits from relatively small key sizes and efficient finite field structure. Furthermore, POK\'E is an IND-CPA secure public-key encryption (PKE) scheme, and it is well known that the IND-CPA PKE can be transformed into an IND-CCA secure KEM via the Fujisaki-Okamoto (FO) transform \cite{fujisaki1999secure}. Although this fact is mentioned in \cite{basso2025poke}, the application of the FO-transform to POK\'E has not been explicitly formalized. In addition, although various isogeny-based key exchange protocols have been proposed in response to key recovery attack, a unified side-by-side comparison of performance is unavailable in the literature.

In this work, we address this gap by performing a comparative analysis of POK\'E-based key exchange and other SIDH variants, evaluating their performance under a unified experimental setting. Specifically, we first explicitly define a KEM derived from POK\'E and analyze its security. Then, we perform a side-by-side evaluation against other isogeny-based key exchange protocols, including terSIDH, M-SIDH, and CSIDH, at the same security level. The contributions of this paper are summarized as follows:

 \begin{itemize}

     \item This paper provides a complete description of a key encapsulation mechanism (KEM) built on POK\'E -- which we denote by POKEx. Although \cite{basso2025poke} provides a high-level description of a POK\'E-based key exchange method, its main focus is the encryption algorithm, and it does not present details of the KEM. This paper supplies that detailed specification. We also provide a C implementation based on \cite{cryptoeprint:2025/1458}. Our implementation of POKEx can be found at:
\begin{center}
    \texttt{https://github.com/gusgkr0117/poke}
\end{center}

     \item This paper provides a comparison of key sizes and performance results across representative isogeny-based key exchange methods—namely, M-SIDH, terSIDH, CSIDH, and a POK\'E-based KEM (POKEx). We evaluate POKEx against these isogeny-based algorithms targeting NIST security level 1. The results show that POKEx is the most efficient, achieving speedups of 21.21 times over terSIDH and 64.97 times over CSIDH. Details of our implementation are provided in Section 4.

 \end{itemize}
\subsection{Organization}
This paper is organized as follows. In Section \ref{sec:pre}, the background knowledge necessary for understanding this paper is provided. Section \ref{sec:pre}, describes the SIDH, Castryck-Decru attack, M-SIDH and bin/terSIDH that counter the Castryck-Decru attack. Section \ref{sec_poks} describes POK\'E and POK\'E-based KEM. Section \ref{sec_im} shows the implementation results, and the conclusion is drawn in Section \ref{sec_conclude}.

\section{Preliminary}
\label{sec:pre}

This section introduces the background knowledge necessary for the paper. This section specifically discusses other isogeny-based key exchange algorithms in comparison with POK\'E. First, SIDH and the Castryck-Decru attack are introduced, followed by an introduction to the M-SIDH and bin/terSIDH that counters the Castryck-Decru attack.

\subsection{SIDH and Castryck-Decru attack}

\subsubsection{Outline of SIDH}
Let $p$ be a prime $p=\ell_A^{e_A}\ell_B^{e_B}\pm 1$ for a relatively prime $\ell_A$ and $\ell_B$, and $\ell_A^{e_A}\approx \ell_B^{e_B}$. Then, select supersingular elliptic curve $E$ over $\mathbb{F}_{p^2}$ having order $(\ell_A^{e_A}\ell_B^{e_B})^2$. Let $\langle P_A, Q_A \rangle$ and $\langle P_B, Q_B \rangle$ be the basis of the torsion subgroup $E[\ell_A^{e_A}]$ and $E[\ell_B^{e_B}]$, respectively. $E[\ell_A^{e_A}]$ is taken as Alice's subgroup, and $E[\ell_B^{e_B}]$ as Bob's subgroup.

Alice computes $\langle R_A \rangle = \langle P_A +[n_A]Q_A\rangle$ using her private key $n_A$. Then, Alice generates an isogeny $\phi_A: E\rightarrow E_A$, where $\ker \phi_A=\langle R_A\rangle$, using the Velu's formula. Alice also computes $\phi_A(P_B)$ and $\phi_A(Q_B)$ and sends $(\phi_A(P_B),\phi_A(Q_B), E_A)$ to Bob. Bob follows a similar process and then sends $(\phi_B(P_A),\phi_B(Q_A), E_B)$ to Alice, where $\phi_B: E\rightarrow E_B=E/\langle R_B\rangle$. Using the points and curve received from Bob, Alice computes the group $\langle R_A' \rangle = \langle \phi_B(P_A)+[n_A]\phi_B(Q_A)\rangle$ and obtains the image curve $E_{AB}=E_B/\langle R_A'\rangle$ through the isogeny from $E_B$ whose kernel is $\langle R_A'\rangle$. Bob follows a similar process and obtained $E_{BA}=E_A/\langle R_B'\rangle$. The shared secret key is the $j$-invariant of the image curves $j(E_{AB})=j(E_{BA})$.

\subsubsection{Castryck-Decru attack on SIDH}

As supersingular elliptic curves are non-commutative, SIDH requires each party to apply their private key to the other party’s public key in order to derive the shared secret. The Castryck–Decru attack exploits this property, combining the leaked torsion-point information with Kani’s glue-and-split algorithm to break SIDH-based schemes in polynomial time.

Let the two subgroups of the elliptic curve $E_0$ whose orders are coprime to each other be $H_1$ and $H_2$. Consider two isogenies $\phi$ and $\gamma$, where $\phi: E_0\rightarrow E_1 = E/H_1$ and $\gamma: E_0\rightarrow E_2 = E/H_2$. Let $\phi': E_2\rightarrow E_3$ be an isogeny with kernel $\gamma(H_1)$ and $\gamma': E_1\rightarrow E_3$ be an isogeny with kernel $\phi(H_2)$ :
\[
    \begin{tikzcd}
    E_0 \arrow[r, "\phi"] \arrow[d, "\gamma"'] & E_1 \arrow[d, "\gamma'" ] \\
    E_2 \arrow[r, "\phi'"'] & E_{3}
    \end{tikzcd}
    \]
Let $N=\#H_0+\#H_1$ and $P_0, Q_0$ be the basis of $E_0[N]$. Define $P_1, Q_1, P_2$ and $Q_2$ as follows:
\begin{align*}
    (P_1, Q_1) &= (\phi(P_0), \phi(Q_0))\\
    (P_2, Q_2) &= (\gamma(P_0), \gamma(Q_0))
\end{align*}
Define the following two-dimensional isogeny:
\begin{align*}
    \rho = \begin{pmatrix}
        \phi & \hat{\gamma}'\\
        -\gamma & \hat{\phi}'
    \end{pmatrix}:E_0\times E_3\rightarrow E_1\times E_2
\end{align*}
Kani prove that $\rho$ is an isogeny with kernel $H=\langle([-d_1 ]P_0,\gamma'\circ\phi(P_0)), ([-d_1]Q_0,\gamma'\circ\phi(Q_0))\rangle$ where $d_1 = \#H_1$. An isogeny between abelian surfaces of this form is called a $(N,N)$-isogeny. The Castryck-Decru attack divides a given isogeny path of smooth degree into several isogeny fragments of small degrees, and guesses each small isogeny with Kani's lemma as a verification tool. If the computation of a two-dimensional isogeny $\rho$ made up of a guessed isogeny path is a product of elliptic curves, then the guessing is correct; however, if it is an unsplit Jacobian curve, then the guessing is incorrect. This leads to a polynomial time attack on SIDH.

\subsection{M-SIDH}

In \cite{fouotsa2023m}, Fouotsa et al. proposed a masking-based countermeasure to mitigate the key recovery attack on SIDH. They proposed M-SIDH and MD-SIDH, which mask the torsion point information and the isogeny degree, respectively. This paper mainly focuses on M-SIDH, which uses a relatively smaller finite field than MD-SIDH for the same security level. 

The core of M-SIDH is to mask torsion point information using a random value $\alpha$. Instead of sending $\phi(P),\phi(Q)$, each party transmits $[\alpha]\phi(P), [\alpha]\phi(Q)$. Let $B$ denote the order of the torsion subgroup of the other party. Since the isogeny degree $deg \phi$ is fixed, the Weil pairing satisfies $e_B([\alpha]P,[\alpha]Q)=e_B(P,Q)^{\alpha^2\deg\phi}$. If $\alpha$ is used such that $\alpha^2\equiv 1 \mod B$, the other party can verify through the Weil pairing whether the masking was applied correctly. To prevent attackers from easily guessing $\alpha$, the parameter $B$ is selected so that the number of $\alpha$ that satisfy $\alpha^2\equiv 1 \mod B$ is sufficiently large.

\subsubsection{Parameter Generation}

Let $\lambda$ be a security level and $t=t(\lambda)$ be an integer determined by $\lambda$. The prime $p$ is of the form $p=ABf-1$, where $A=\prod_{i=1}^t\ell_i$ and $B=\prod_{i=1}^t q_i$. The $\ell_i$s and $q_i$s are all distinct prime numbers, $A, B$ are relatively prime, $A\approx B\approx\sqrt{p}$, and $f$ is the cofactor. For a supersingular elliptic curve $E$ defined over $\mathbb{F}_{p^2}$, let $\langle P_A, Q_A \rangle$ and $\langle P_B, Q_B\rangle$ be the basis of $E[A]$ and $E[B]$, respectively. Define $\mu_2(N)$ as follows:
\begin{align}
    \mu_2(N)=\{x\in \mathbb{Z}/N\mathbb{Z} \mid x^2\equiv 1 \mod N\}
\label{masked}
\end{align}

\subsubsection{Public Key Generation}
Alice selects random integers $\alpha$ and $a$ in $\mu_2(B)$ and $\mathbb{Z}/A\mathbb{Z}$, respectively. Then, Alice computes an isogeny $\phi_A:E\rightarrow E_A$, where $\ker\phi_A = \langle P_A+[a]Q_A\rangle$. Alice sends $(E_A, [\alpha]\phi_A(P_B),[\alpha]\phi_A(Q_B))$ to Bob. Similarly, Bob selects random integers $\beta$ and $b$ in $\mu_2(A)$ and $\mathbb{Z}/B\mathbb{Z}$, respectively. Then, Bob computes an isogeny $\phi_B:E\rightarrow E_B$, where $\ker\phi_B=\langle P_B+[b]Q_B\rangle$, and sends $(E_B,[\beta]\phi_B(P_A),[\beta]\phi_B(Q_A))$ to Alice.

\subsubsection{Shared Secret}
Let $(E_B,R_A,S_A)$ be the tuple Alice receives. Alice checks $e_A(R_A,S_A)$ equals $e_A(P_A,Q_A)^B$. If not, the algorithm terminates. Then, Alice computes an isogeny $\phi_A':E_B\rightarrow E_{AB}=E_B/\langle R_A+[a]S_A\rangle$. Bob also computes $e_B(R_B,S_B)$ and checks $e_B(R_B,S_B)$ equals $e_B(P_B,Q_B)^B$ using the information $(E_A,R_B,S_B)$ received from Alice. Then, Bob computes an isogney $\phi_B':E_A\rightarrow E_{BA}=E_A/\langle R_B+[b]S_B\rangle$. The shared secret between Alice and Bob is $j(E_{AB})=j(E_{BA})$.

\subsection{bin/terSIDH}

In \cite{basso2023new}, they present binSIDH (binary SIDH) and terSIDH (ternary SIDH) as countermeasures against the Castryck–Decru attack. Both introduce the concept of artificial orientation and adopt new forms of secret keys. Among them, terSIDH provides about 1.6 bits of security per bit of the private key and offers reasonably fast execution. For this reason, we focus on terSIDH in this paper. 

An artificial $A$-orientation of a supersingular elliptic curve $E$ over $\mathbb{F}_{p^2}$ is defined as a pair of cyclic subgroups $\mathfrak{A}=(G_1, G_2)$ of order $A$, where $G_1, G_2 \subset E[A]$ satisfying $G_1 \cap G_2 = \{0\}$. The pair $(E, \mathfrak{A})$ is referred to as an artificial $A$-oriented curve. This notion enables the construction of isogenies whose kernels are direct sums of subgroups taken from $G_1$ and $G_2$, that is, kernels of the form $H_1 \oplus H_2$ with $H_i \subseteq G_i$.

\subsubsection{Setup}
Let $\lambda$ be a security level and $t=t(\lambda)$ be an integer determined by $\lambda$. The prime $p$ is of the form $p=ABf-1$, where $A=\prod_{i=1}^t\ell_i$ and $B=\prod_{i=1}^t q_i$. The $\ell_i$s and $q_i$s are all distinct prime numbers, $A, B$ are relatively prime, $A\approx B\approx\sqrt{p}$, and $f$ is the cofactor. Let $E_0$ be a supersingular elliptic curve defined over $\mathbb{F}_{p^2}$, $\mathfrak{A}$ be an artificial $A$-orientation on $E_0$ and let $\mathfrak{B}$ be an artificial $B$-orientation on $E_0$. The public parameters are $E_0, p, A, B, \mathfrak{A}$, and $\mathfrak{B}$.

\subsubsection{Public Key generation}
Alice samples a secret vector $\vec{a}$ from $\{0,1,2\}^t$, uniformly at random. Then, Alice computes the $\mathfrak{A}$-oriented isogeny $\phi_A:E_0\rightarrow E_A$ defined by $\vec{a}$, where $\deg \phi_A \mid A$. Alice computes the push-forward $\mathfrak{B}'$ of $\mathfrak{B}$ on $E_A$ through $\phi_A$ and sends $(E_A, \mathfrak{B}')$ to Bob. Bob follows a similar process and sends $(E_B, \mathfrak{A}')$ to Alice with his secret vector  $\vec{b} \in \{0,1,2\}^t$.

\subsubsection{Shared secret}
Upon receiving Bob's public key $(E_B, \mathfrak{A}')$, Alice checks that $\mathfrak{A}'$ is an artificial $A$-orientation on $E_B$. If not, she aborts. Then, Alice computes the $\mathfrak{A}'$-oriented isogeny $\phi_A':E_B\rightarrow E_{BA}$ of degree $A$ defined by $\vec{a}$. Bob follows a similar process and computes the $\mathfrak{B}'$-oriented isogeny $\phi_B':E_A\rightarrow E_{AB}$. The shared secret key between Alice and Bob is $j(E_{AB})=j(E_{BA})$.

\section{POK\'E}
\label{sec_poks}
In \cite{dartois2024thetamodel}, the computation of a $(2,2)$-isogeny becomes very fast based on the theta model, which leads to the two-dimensional representation of an isogeny with an arbitrary degree. This randomized degree makes the Castryck-Decru attack infeasible. In \cite{basso2025poke}, Basso proposed a fast key exchange scheme by combining the SIDH structure and the two-dimensional isogeny representation.

In this section, we present the POK\'E encryption algorithm and the key exchange algorithm derived from it. Although the name POK\'E originates from \textit{POint-based Key Exchange}, \cite{basso2025poke} does not explicitly describe a key exchange algorithm. Since a key exchange algorithm is as important as the performance analysis, we first explain the encryption process given in \cite{basso2025poke}, and then describe the key exchange algorithm based on it.

POK\'E uses a prime $p$ of a form $p = 2^a 3^b 5^c f -1$ such that $a \approx \lambda$, $3^b \approx 2^{\lambda/2}$ and $5^c \approx 2^{\lambda/3}$. Let $E_0 : y^2 =x^3 + x$ of $j$-invariant 1728 as the starting curve with known endomorphism ring. Let $(P_0, Q_0)$ be a $2^a$-torsion basis on $E_0$, $(R_0, S_0)$ be a $3^b$-torsion basis and $(X_0, Y_0)$ be a $5^c$-torsion basis.

\subsection{POK\'E Encryption algorithm }
\subsubsection{Key Generation}
Here, we adopt the heuristic version of the keygen algorithm in \cite{basso2025poke}. Alice computes an endomorphism $\theta \in \mathrm{End}(E_0)$ of degree $3^b \cdot q(2^a-q)$ for some uniformly random value $q \in [0, 2^a)$. Let $\theta = \hat{\phi} \circ \psi$ where $\deg\phi=3^b$ and $\deg\psi = q(2^a-q)$. Alice can computes the $\phi : E_0 \to E_A$ by evaluating its kernel $\ker\phi = \ker\theta \cap E_0[3^b]$. Since $\psi = [3^{-b}]\circ\hat{\phi}\circ \theta$, she can compute $\psi$. Alice chooses uniformly random scalars $\alpha, \beta \in [0,2^a)$, $\gamma\in [0, 3^b)$ and $\delta_5 \in [0, 5^c)$. Alice's public key is
\begin{align*}
    E_A, (P_A, Q_A) &= ([\alpha]\psi(P_0), [\beta]\psi(Q_0)),\\
    (R_A, S_A) &= ([\gamma]\psi(R_0), [\gamma]\psi(S_0)),\\
    (X_A, Y_A) &= ([\delta_5]\psi(X_0), [\delta_5]\psi(Y_0)).
\end{align*}
Note that an attacker cannot recover the random scalars $\gamma$ and $\delta_5$ since the degree of $\psi$ is randomized. Alice's secret key is
$$
    \alpha, \beta, \delta_5, q.
$$
\subsubsection{Encryption}
Bob chooses a uniformly random scalar $\delta \in [0, 3^b)$ and computes isogenies $\phi_B : E_0\to E_B$ and $\phi_B' : E_A \to E_{AB}$ such that $\ker\phi_B = \langle R_0 + [\delta] S_0\rangle$ and $\ker\phi_B' = \langle R_A + [\delta] S_A\rangle$. For a uniformly random matrix $D_5 \in \mathrm{SL}_2(\mathbb{Z}/5^c\mathbb{Z})$ and a uniformly random scalar $\omega \in [0, 2^a)$, Bob computes
\begin{align*}
    (P_B, Q_B) &= ([\omega]\phi_B(P_0), [\omega^{-1}]\phi_B(Q_0))\\
    (P_{AB}, Q_{AB}) &= ([\omega]\phi_B'(P_A), [\omega^{-1}]\phi_B'(Q_A))\\
    (X_B, Y_B) &= D_5 \cdot (\phi_B(X_0), \phi_B(Y_0))\\
    (X_{AB}, Y_{AB}) &= D_5\cdot (\phi_B'(X_A), \phi_B'(Y_A)).
\end{align*}
The ciphertext is
\begin{align*}
    &E_B, (P_B, Q_B), (X_B, Y_B), E_{AB}, (P_{AB}, Q_{AB}),\text{ and }\\
    &ct = m\oplus \mathrm{KDF}(X_{AB}, Y_{AB}),
\end{align*}
where  KDF is a key derivation function.
\subsubsection{Decryption}
Alice recovers the points $(X_{AB}, Y_{AB})$ and computes $m$ from $ct$. The secrets $q, \alpha$ and $\beta$ are used to compute the isogeny $\psi'$ and $\delta_5$ is used to unmask the given torsion basis. Alice computes the message $m$ as
\begin{align*}
    (X_{AB}, Y_{AB}) &= ([\delta_5]\psi'(X_B), [\delta_5]\psi'(Y_B))\\
    m &= ct\oplus \mathrm{KDF}(X_{AB}, Y_{AB}),
\end{align*}
where KDF is the key derivation function used in the encryption algorithm and $\psi' : E_B \to E_{AB}$ with a two-dimensional representation $\ker\Psi' = \langle(P_B,[\alpha^{-1}]P_{AB}), (Q_B, [\beta^{-1}]Q_{AB}) \rangle$. The overall process of POK\'E protocol is described as a diagram in Figure \ref{figure: the poke diagram}.

\begin{figure}[H]
\centering
\begin{tikzpicture}[
  >=Latex,
  font=\small,
  every node/.style={align=left},
  box/.style={draw=none, inner sep=1pt},
  arrow/.style={->, thick},
  red/.style={text=red},
  orange/.style={text=orange},
  blue/.style={text=blue},
  lblue/.style={text=cyan},
  purple/.style={text=purple}
  ]

\node (E0) at (0, 2.5) {$E_0$};
\node (EA) at (5, 2.5) {$\textcolor{orange}{E_A}$};
\node (EB) at (0, 0) {$\textcolor{cyan}{E_B}$};
\node (EAB) at (5, 0) {$\textcolor{cyan}{E_{AB}}$};

\draw[arrow, draw=red] (E0) -- node[below, red]{$\psi$, $\deg\,\psi = q(2^a-q)$} (EA);
\draw[arrow, draw=blue] (E0) -- node[left, blue]{$\phi_B$} (EB);
\draw[arrow, draw=blue] (EA) -- node[right, blue]{$\phi_B'$} (EAB);
\draw[arrow, draw=red] (EB) -- node[below, red]{$\psi'$} (EAB);

\node[box, anchor=north east] at ($(E0)+(1,1.5)$) {
  $\begin{aligned}
  &P_0, Q_0 \in E_0[2^a]\\
  &R_0, S_0 \in E_0[3^b]\\
  &X_0, Y_0 \in E_0[5^c]
  \end{aligned}$
};

\node[box, anchor=north west] at ($(EA)+(-2.2,1.5)$) {
  $\begin{aligned}
  &\textcolor{orange}{P_A, Q_A} = [\textcolor{red}{\alpha}]\textcolor{red}{\phi}(P_0), [\textcolor{red}{\beta}]\textcolor{red}{\phi}(Q_0)\\
  &\textcolor{orange}{R_A, S_A} = [\textcolor{red}{\gamma}]\textcolor{red}{\phi}(R_0), [\textcolor{red}{\gamma}]\textcolor{red}{\phi}(S_0)\\
  &\textcolor{orange}{X_A, Y_A} = [\textcolor{red}{\delta_5}]\textcolor{red}{\phi}(X_0), [\textcolor{red}{\delta_5}]\textcolor{red}{\phi}(Y_0)
  \end{aligned}$
};

\node[box, anchor=south east] at ($(EB)+(2,-1.9)$) {
  $\begin{aligned}
  &\textcolor{cyan}{P_B, Q_B} \\
  &= [\textcolor{blue}{\omega}]\textcolor{blue}{\phi_B}(P_0), [\textcolor{blue}{1/\omega}]\textcolor{blue}{\phi_B}(Q_0)\\
  &\textcolor{cyan}{X_B, Y_B} \\
  &= \textcolor{blue}{D_5}\cdot (\textcolor{blue}{\phi_B}(X_0), \textcolor{blue}{\phi_B}(Y_0))
  \end{aligned}$
};

\node[box, anchor=south west] at ($(EAB)+(-2,-2.5)$) {
  $\begin{aligned}
  &\textcolor{cyan}{P_{AB}, Q_{AB}}\\
  &= [\textcolor{blue}{\omega}]\textcolor{blue}{\phi_B'}(\textcolor{orange}{P_A}), [\textcolor{blue}{1/\omega}]\textcolor{blue}{\phi_B'}(\textcolor{orange}{Q_A})\\
  &\textcolor{purple}{X_{AB}, Y_{AB}} \\
  &= \textcolor{blue}{D_5} \cdot (\textcolor{blue}{\phi_B'}(\textcolor{orange}{X_A}), \textcolor{blue}{\phi_B'}(\textcolor{orange}{Y_A}))\\
  &= [\textcolor{red}{\delta_5}]\textcolor{red}{\psi'}(\textcolor{cyan}{X_B}), [\textcolor{red}{\delta_5}]\textcolor{red}{\psi'}(\textcolor{cyan}{Y_B})
  \end{aligned}$
};

\end{tikzpicture}
\caption{The POK\'E protocol. Public parameters are in black; elements in \textcolor{red}{red} are part of the secret key or computed during decryption; \textcolor{orange}{orange} parts are public keys; \textcolor{blue}{blue} are ephemeral secrets; \textcolor{cyan}{light blue} are ciphertext parts; \textcolor{purple}{purple} shows the shared secret.}
\label{figure: the poke diagram}

\end{figure}
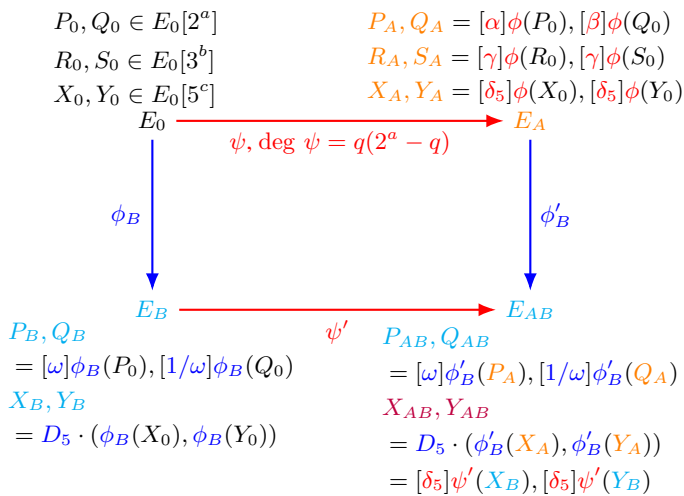

\begin{remark}
 POK\'E withstands the key recovery attack against SIDH. As discussed previously, the key recovery attack requires the following conditions: (1) access to the torsion-point information published by SIDH-based cryptographic schemes, (2) knowledge of the degree of the isogeny connecting two elliptic curves, and (3) that this degree is power-smooth to enable an efficient attack. As can be seen in the algorithm description, the publicly revealed torsion points are masked by random values, thereby addressing condition (1). Furthermore, by employing isogenies of arbitrary degree, the scheme mitigates conditions (2) and (3). In addition, since the isogeny degree is both unknown and sufficiently large, POKÉ also resists the backdoor attack against FESTA described in \cite{castryck2023polynomial}.
\end{remark}


\subsection{POK\'E Key Encapsulation Mechanism}
As the author mentioned in \cite{basso2025poke}, POK\'E public key encryption scheme can be converted to IND-CCA KEM by applying the Fujisaki-Okamoto (FO) transform \cite{fujisaki1999secure}. However, the paper does not present an explicit construction. In this work, we provide a complete FO-based KEM design, including the full algorithms, which we refer to as \textbf{POKEx}.

\begin{algorithm}[ht]
\caption{POKEx Key Genearation}
\begin{algorithmic}[1]
\vskip 1mm
\REQUIRE The starting curve $E_0$ and the bases $\langle P_0, Q_0\rangle= E_0[2^a]$, $\langle R_0, S_0 \rangle = E_0[3^b]$, and $\langle X_0, Y_0 \rangle \in E_0[5^c]$
\ENSURE $(pk, sk)$ 
\STATE Sample a random $q\in[0,2^a)$ such that $q(2^a-q)$ is coprime with 2, 3, and 5.
\STATE Generate a $q(2^a-q)$-isogeny $\phi:E_0\rightarrow E_A$.
\STATE Sample a random $\alpha_2, \beta_2 \xleftarrow{\$}\mathbb{Z}_{2^a}^{\times}$.
\STATE Sample a random $\gamma_3 \xleftarrow{\$} \mathbb{Z}_{3^b}^{\times}$.
\STATE Sample a random $\delta_5 \xleftarrow{\$}\mathbb{Z}_{5^c}^{\times}$.
\STATE Compute $P_A, Q_A = [\alpha_2]\phi(P_0), [\beta_2]\phi(Q_0)$.
\STATE Compute $R_A, S_A = [\gamma_3]\phi(S_0), [\gamma_3]\phi(R_0)$.
\STATE Compute $X_A, Y_A = [\delta_5]\phi(X_0), [\delta_5]\phi(Y_0)$.
\STATE Sample a random $s \xleftarrow{\$} \{0,1\}^{256}$.
\STATE Return $sk=(q, \alpha_2, \beta_2, \delta_5,s)$ and \newline $pk=(E_A, P_A, Q_A, R_A, S_A, X_A, Y_A)$.
\end{algorithmic}
\label{alg:pokex_keygen}
\end{algorithm}

\begin{algorithm}[ht]
\caption{POKEx Key Encapsulation}
\begin{algorithmic}[1]
\vskip 1mm
\REQUIRE A public key $pk$.
\ENSURE A shared key $K$ and a ciphertext $ct$.
\STATE Sample a random message $m \xleftarrow{\$} \{0,1\}^{256}$.
\STATE $ct \leftarrow \mathsf{POKE.ENC}(pk,m;G_{\lambda}(m))$.
\STATE Return $K = H_{\lambda}(m,ct)$.
\end{algorithmic}
\label{alg:pokex_encap}
\end{algorithm}

\begin{algorithm}[h]
\caption{POKEx Key Decapsulation}
\begin{algorithmic}[1]
\vskip 1mm
\REQUIRE A ciphertext $ct$, a public key $pk$, and a secret key $sk$
\ENSURE A shared key $K$
\STATE $m\leftarrow \mathsf{POKE.DEC}(sk,ct)$.
\STATE Compute $ct' \leftarrow \mathsf{POKE.ENC}(pk,m;G_{\lambda}(m))$.
\IF {$ct'=ct$}  $K\leftarrow H_{\lambda}(m,ct)$.
\ELSE \ $K\leftarrow H_{\lambda}(s,ct)$.
\ENDIF
\STATE Return $K$.
\end{algorithmic}
\label{alg:pokex_decap}
\end{algorithm}

\subsubsection{Setup and Key Generation}
Let $G = \{G_\lambda\}$ and $H = \{H_\lambda\}$ be families of cryptographic hash functions such that $G_\lambda : \mathcal{M}_\lambda \to \mathcal{R}_\lambda$ and $H_\lambda : \{0,1\}^* \to \mathcal{M}_\lambda$, where $\mathcal{M}_\lambda$ denotes the message space and $\mathcal{R}_\lambda$ is the randomness space for the POK\'E encryption.
For key generation, we compute $(pk, sk)$ from the POK\'E key generation algorithm. We then choose a uniformly random dummy message $s$, and define the secret key as the pair $(sk, s)$, while the public key is $pk$.

\subsubsection{Key Encapsulation}
Let $m\in \mathcal{M}_\lambda$ be a uniformly random message. The ciphertext is computed as $ct \leftarrow \mathrm{Enc}(pk, m; G_\lambda(m))$, where $\mathrm{Enc}$ denotes the POK\'E encryption algorithm, and $G_\lambda(m)$ provides the randomness used in the encryption. The shared key is derived as $K \leftarrow H_\lambda(m,ct)$.

\subsubsection{Key Decapsulation}
We extract $m \leftarrow \mathrm{Dec}(sk, ct)$ and compute $ct' \leftarrow \mathrm{Enc}(pk, m ;G(m))$. If $ct' = ct$, we compute the key $K \leftarrow H_\lambda(m, ct)$. Otherwise, we compute $K \leftarrow H_\lambda(s,ct)$.

The following Algorithms \ref{alg:pokex_keygen}, \ref{alg:pokex_encap}, and \ref{alg:pokex_decap} correspond to the key generation (setup), Encapsulation, and Decapsulation procedures of POKEx, respectively. In Algorithm \ref{alg:pokex_keygen}, Step 2 follows Algorithm 3 of \cite{basso2025poke} in our implementation. Within the algorithms, $\mathsf{POKE.ENC}$ denotes the POK\'E encryption (Algorithm 4 in \cite{basso2025poke}) and $\mathsf{POKE.DEC}$ denotes the POK\'E decryption (Algorithm 5 in \cite{basso2025poke}). The Figure \ref{fig:pokekem} depicts the POKEx algorithm.

\begin{figure*}[t]
\centering
\caption{POKEx: POK\'E-based KEM}\label{fig:pokekem}
\begin{tabular}{|lcl|}
\hline
$\mathsf{KEM.KeyGen()}$& &\\
~~~ $q \xleftarrow{\$} [0,2^a)$& &\\
~~~ $\alpha_2, \beta_2 \xleftarrow{\$}\mathbb{Z}_{2^a}^{\times}$& &\\
~~~  $\gamma_3 \xleftarrow{\$} \mathbb{Z}_{3^b}^{\times}$& &\\
~~~  $\delta_5 \xleftarrow{\$}\mathbb{Z}_{5^c}^{\times}$& &\\
~~~  $P_A, Q_A \leftarrow [\alpha_2]\phi(P_0), [\beta_2]\phi(Q_0)$& &\\
~~~  $R_A, S_A \leftarrow [\gamma_3]\phi(S_0), [\gamma_3]\phi(R_0)$& &\\
~~~   $X_A, Y_A \leftarrow [\delta_5]\phi(X_0), [\delta_5]\phi(Y_0)$& &\\
~~~   $s \xleftarrow{\$} \{0,1\}^{256}$& &\\
~~~   $sk=(q, \alpha_2, \beta_2, \delta_5,s)$& & $\mathsf{KEM.Encap}(pk)$\\
~~~ $pk=(E_A, P_A, Q_A, R_A, S_A, X_A, Y_A)$ &
~~~ $\begin{CD} @>~~~~~~~{\qquad pk \qquad}~~~~~~>>\end{CD}$
~~~ & $m \xleftarrow{\$} \{0,1\}^{256}$\\
~~~ && $ct \leftarrow \mathsf{POKE.ENC}(pk,m;G_{\lambda}(m))$\\
$\mathsf{KEM.Decap}(ct,pk,sk)$ && $K \leftarrow H_{\lambda}(m,ct)$\\
~~~ $m\leftarrow \mathsf{POKE.DEC}(sk,ct)$&$\begin{CD} @<~~~~~~~{\qquad ~~~ ct \qquad}~~~~~~<<\end{CD}$& return $(K,ct)$\\
~~~ $ct' \leftarrow \mathsf{POKE.ENC}(pk,m;G_{\lambda}(m))$ &&\\
~~~ if $ct'=ct$ &&\\
~~~ then return $K\leftarrow H_{\lambda}(m,ct)$ &&\\
~~~ else return $K\leftarrow H_{\lambda}(s,ct)$&&\\
\hline
\end{tabular}
\end{figure*}

\subsubsection{Security of POKEx}

We briefly describe the security of POKEx. The outline is as follows. The security of POK\'E PKE is based on the hardness of the C-POKE problem in \cite{basso2025poke}, and according to \cite{basso2025poke}, it is IND-CPA secure PKE in the random oracle model under the assumption that the C-POKE problem is hard. By adopting the FO-transform, it is well-known that an IND-CCA secure KEM can be constructed from an IND-CPA secure PKE, which has been proven in \cite{fujisaki1999secure}. This application has already been mentioned in \cite{basso2025poke}, and that is what we call POKEx. For the security of POKEx, we only restate the security assumption and the relevant theorem on which the security of POK\'E PKE is based.

\begin{problem}[C-POKE]\cite{basso2025poke}
Let \( E_0 \) be a supersingular elliptic curve defined over \( \mathbb{F}_{p^2} \), with known endomorphism ring. Let \( P_0, Q_0 \) be a basis of \( E_0[2^a] \), \( R_0, S_0 \) a basis of \( E_0[3^b] \), and \( X_0, Y_0 \) a basis of \( E_0[5^c] \).

Let \( \phi: E_0 \rightarrow E_A \) be an isogeny of degree \( q(2^a - q) \), for some unknown value \( q \in [0, 2^a] \). Write
\[
\begin{aligned}
  P_A, Q_A &= \left[ \alpha_2 \right] \phi(P_0), \quad \left[ \beta_2 \right] \phi(Q_0), \\
  R_A, S_A &= \left[ \gamma_3 \right] \phi(R_0), \quad \left[ \gamma_3 \right] \phi(S_0), \\
  X_A, Y_A &= \left[ \delta_5 \right] \phi(X_0), \quad \left[ \delta_5 \right] \phi(Y_0),
\end{aligned}
\]
where \( \alpha_2, \beta_2, \gamma_3, \delta_5 \) are uniformly random scalars from
\[
\mathbb{Z}_{2^a}^\times \times \mathbb{Z}_{2^a}^\times \times \mathbb{Z}_{3^b}^\times \times \mathbb{Z}_{5^c}^\times.
\]

Let \( \psi: E_0 \rightarrow E_B \) be an isogeny of degree \( 3^b \), and write \( \psi': E_A \rightarrow E_{AB} \) for its pushforward \( \phi_\ast \psi \). Write
\[
\begin{aligned}
  P_B, Q_B &= \left[ \omega_2 \right] \psi(P_0), \quad \left[ 1/\omega_2 \right] \psi(Q_0), \\
  P_{AB}, Q_{AB} &= \left[ \omega_2 \right] \psi'(P_A), \quad \left[ 1/\omega_2 \right] \psi'(Q_A), \\
  X_B, Y_B &= \mathbf{D}_5 \begin{bmatrix} \psi(X_0), & \psi(Y_0) \end{bmatrix}^T, \\
  X_{AB}, Y_{AB} &= \mathbf{D}_5 \begin{bmatrix} \psi'(X_A), & \psi'(Y_A) \end{bmatrix}^T,
\end{aligned}
\]
where \( \omega_2 \) is a uniformly random scalar from \( \mathbb{Z}_{2^a}^\times \), and \( \mathbf{D}_5 \) is a uniformly random matrix from \( \mathrm{SL}_2(\mathbb{Z}_{5^c}) \).
Given 
\begin{align*}
    &\left(E_0, (P_0, Q_0), (R_0, S_0), (X_0, Y_0) \right),\\
    &\left(E_A, (P_A, Q_A), (R_A, S_A), (X_A, Y_A) \right),\\
    &\left(E_B, (P_B, Q_B), (X_B, Y_B) \right),\\
    &\left(E_{AB}, (P_{AB}, Q_{AB}) \right),
\end{align*}
compute the points \( X_{AB}, Y_{AB} \).
\end{problem}

\begin{theorem}\cite{basso2025poke}
The POK\'E protocol, where the key-derivation function KDF is
modeled as a random oracle, is IND-CPA secure in the random oracle model
under the assumption that the C-POKE problem is hard.
\end{theorem}

\begin{proof}
    The proof follows directly from \cite[Theorem~8]{basso2025poke}.
\end{proof}

\begin{theorem}\cite{hofheinz2017modular}
For any IND-CCA adversary B against $\mathrm{POKEx}$, issuing at most $q_G$ (resp. $q_H$) queries to the random oracle $G$ (resp. $H$), there exists an IND-CPA adversary $A$ against POK\'E $\mathrm{PKE}$ with
\begin{align*}
    \mathsf{Adv}_{\mathsf{POKEx}}^{\mathsf{IND\text{-}CCA}}(B)\leq \frac{2q_G+q_H+1}{2^n}+3\cdot \mathsf{Adv}_{\mathsf{POKE\text{-}PKE}}^{\mathsf{IND\text{-}CPA}}(A),
\end{align*}
where the size of the message space of $\mathrm{PKE}$ is $2^n$.
\end{theorem}

\begin{proof}
    POK\'E PKE is perfectly correct \cite{idris2025transforming}. Hence, the proof follows directly from combining the results from Theorem 3.2 and Theorem 3.4 in \cite{hofheinz2017modular}.
\end{proof}

\section{Implementation Results}
\label{sec_im}
In this section, we present the comparison result of M-SIDH, terSIDH, and POK\'E. We selected MSIDH, terSIDH, and CSIDH as comparison baselines for POKEx for the following reasons. MSIDH is regarded as more efficient than MD-SIDH among masked variants designed to counter the SIDH attacks, and terSIDH is considered the most efficient among alternative approaches that mitigate those attacks by introducing artificial orientation. Lastly, CSIDH is a CRS-based isogeny key exchange algorithm that is not affected by SIDH key-recovery attacks. Except for MSIDH, terSIDH, POKEx, and CSIDH use parameters targeting NIST security level 1. 

\subsection{Parameter selection}
\subsubsection{Parameter selection for MSIDH}
The parameters of MSIDH follow those specified in \cite{kim2024performance}. To follow NIST security level 1, MSIDH requires a large finite field of at least 5911 bits\cite{fouotsa2023m}. However, as denoted in \cite{kim2024performance}, MSIDH still performs significantly slower than other SIDH variants. Therefore, \cite{kim2024performance} also used the 4096-bit parameters mainly to provide a rough performance estimate. Following the same approach, we also adopt the MSIDH-4096 parameters proposed in \cite{kim2024performance}, targeting the classical 94-bit security level.

For $\texttt{MSIDH-4096}$, the following 4095-bit prime number is used.
\begin{align*}
    p_{}=2^2\ell_1 \ell_2 \cdots \ell_{417}-1
\end{align*}
Here, $\ell_1,\cdots\ell_{416}$ denote the first 228 odd prime numbers, and $\ell_{417}=2897$. Alice’s torsion subgroup has order $A=2^2\cdot \ell_1\ell_3\cdots\ell_{417}$, while Bob’s torsion subgroup has order $B=\ell_2\ell_4\cdots\ell_{416}$. This choice of primes ensures that the subgroup orders of Alice and Bob are comparable in size and that the degrees of the corresponding isogenies align appropriately. The resulting parameters provide a classical security level of 94 bits and an estimated quantum security level of about 47 bits.

Using this prime, M-SIDH uses the following curve as the starting curve defined over $\mathbb{F}_{p^2} = \mathbb{F}_p[i]$, where $i^2 = -1$:
\begin{align*}
E: y^2 = x^3 + M_{4095}\cdot i x^2 + x.
\end{align*}
The exact coefficients are available in \cite{kim2024performance}.

\subsubsection{Parameter selection for terSIDH}
Targeting NIST security level 1, the following 1570-bit prime number is used.
\begin{align*}
    p_{}&=4\cdot 363\cdot\ell_1\ell_2\cdots\ell_{185}-1\\
    A&=4\prod_{i=1}^{92}\ell_{2i}\\
    B&=\prod_{i=1}^{93}\ell_{2i-1}
\end{align*}
where, $\ell_1\cdots \ell_{185}$ denote the first 185 odd prime numbers. 
Using this prime, terSIDH uses the following curve as the starting curve defined over $\mathbb{F}_{p^2}$:
\begin{align*}
E: y^2 = x^3 + 6 x^2 + x.
\end{align*}

\subsubsection{Parameter selection for POKEx}
For NIST security level 1, it uses a 431-bit prime number such as
$$
    p = 2^{129}\cdot3^{164}\cdot 5^{18} - 1.
$$
The starting curve is $E_0 : y^2 = x^3 + x$ defined over $\mathbb{F}_{p^2}$ such that $j(E_0) = 1728$.

\subsection{Performance Results}

\begin{table*}[t]
\centering
\caption{Public key and secret key sizes targeting NIST security level 1}
\label{tab:keysize}
\vspace{0.5em}
\setlength{\tabcolsep}{4pt}       
\begin{tabular}{@{}lcccc:c:cc@{}}
\toprule
 & \textbf{M-SIDH} & \textbf{terSIDH} & \textbf{POKEx} & \textbf{CSIDH-4096} & \textbf{ML-KEM-512} & \textbf{ECDH (secp256r1)} & \textbf{DH (MODP-2048)}\\
\midrule
$\log_2 p$ 
  & 5911  & 1570  & 431 & 4095 & - & 256 & 2048 \\
\textbf{pk size (B)} 
  & 4434  &  1178 & 324 & 512 & 800 & 65 & 256\\
\textbf{sk size (B)} 
  & 369 & 24 &  54 & 35 & 1632&  32 & 256\\
\bottomrule
\end{tabular}
\end{table*}

  \begin{table*}[t]
\centering
\caption{Performance comparison of key exchange algorithms in C}
\label{tab:C}
\centering
\begin{tabular}{ccccc:c:cc}
\thickhline
                  &     \texttt{MSIDH-4096}  & terSIDH & CSIDH-4096 & POKEx  & ML-KEM-512 &ECDH(secp256r1) &DH (MODP-2048)\\ \thickhline
Classic Security        &       94 & 128 & 128 & 128  & 128 & 128&128\\ 
Key exchange (ms) & 189,769.19 & 6,510.83 & 19,941.20 &306.95  &  0.10 & 1.43 &  2.26 \\
\thickhline
\end{tabular}
\end{table*}

First, we provide the uncompressed versions of public key and secret key sizes of M-SIDH, terSIDH, POKEx, and CSIDH at NIST security level 1. The CPU used for the measurements was an Intel Core i9-10980XE with a clock speed of 3.0 GHz, and the tests were conducted on an Ubuntu 20.04 LTS operating system using the optimization level -O3 and the gcc 9.4.0 compiler.

Also, for comparison with standardized algorithms from the NIST PQC competition as well as representative pre-quantum key establishment mechanisms, we selected ML-KEM, ECDH, and DH. For the post-quantum setting, we chose ML-KEM-512, which achieves NIST security level 1 and corresponds to approximately 128-bit classical security. For pre-quantum algorithms, ECDH over the NIST P-256 (secp256r1) curve and DH-2048 were employed, both of which are commonly regarded as providing comparable 128-bit classical security.

All implementations used in this work are reference implementations. For ML-KEM, we adopted the official Round 3 reference implementation provided by the NIST PQC project. For ECDH and DH, we used the \texttt{wolfSSL} cryptographic library \cite{wolfssl}. In particular, for finite-field Diffie–Hellman, we used the standardized MODP-2048 group (RFC 3526, Group 14), where public keys and shared secrets have a fixed size of 256 bytes.

In Table~\ref{tab:C}, \textbf{Key exchange} denotes the total time required for Alice and Bob to establish a shared key. Specifically, \textbf{Key exchange} is the total time of the following four steps: 1) Alice generates her public value from her private key and sends it to Bob, 2) Bob generates his public value from his private key and sends it to Alice, 3) Alice computes the shared secret using Bob’s public value, and 4) Bob computes the shared secret using Alice’s public value. For the KEM algorithm, \textbf{Key exchange} is the total time of 1) Key generation, 2) Encapsulation, and 3) Decapsulation.

As shown in Table \ref{tab:C}, POKEx is the most efficient key exchange method among isogeny-based algorithms. First, MSIDH offers a lower security strength than the other algorithms and exhibits the slowest performance, indicating that it is impractical for real-world implementations. POKEx is about 64.97 times faster than CSIDH and about 21.21 times faster than terSIDH. These results show that POKEx currently provides the most efficient key exchange algorithm among isogeny-based cryptosystems.

Nevertheless, when compared with other post-quantum cryptographic schemes and classical pre-quantum algorithms such as ECC- and DH-based key exchange, further performance improvements of POKEx are still required. In particular, compared to ML-KEM, POKEx exhibits considerably slower performance. However, as discussed in \cite{stebila2024security}, ML-KEM suffers from a notable drawback: its substantially larger public keys and ciphertexts introduce non-negligible bandwidth overhead compared to ECDH, which has direct implications for practical protocol design. Recent TLS-level evaluations further corroborate this observation, showing that while ML-KEM offers excellent computational performance, it incurs a substantial bandwidth overhead due to its large public keys and ciphertexts, highlighting that rather than computation, communication cost remains a critical bottleneck in post-quantum protocol design \cite{montenegro2025performance}.

In this respect, POKEx offers an important advantage that its key sizes are significantly smaller than those of ML-KEM, resulting in reduced communication overhead. This bandwidth efficiency suggests that POKEx has considerable potential for further optimization and may serve as a competitive alternative in scenarios where both post-quantum security and communication efficiency are critical.

\section{Conclusion}
\label{sec_conclude}

This paper explicitly specified a POK\'E-based key encapsulation mechanism -- POKEx -- and evaluated its performance against representative isogeny-based key exchange algorithms at the same security level. Our comparison algorithms consist of M-SIDH, the most efficient masking-based countermeasure against SIDH key recovery attacks; terSIDH, the most efficient SIDH-variant with conceptually different countermeasure; and CSIDH, a CRS-based scheme that is unaffected by SIDH key recovery attacks.

Targeting NIST security level 1, M-SIDH exhibited the slowest performance. POKEx outperformed terSIDH and CSIDH by factors of 21.21 and 64.97, respectively. These results show that POKEx is currently the most efficient isogeny-based key exchange algorithm.

On the other hand, when compared with other post-quantum cryptographic schemes and classical pre-quantum algorithms such as ECC- and DH-based key exchange, further performance improvements of POKEx are still required. Nevertheless, POKEx offers an important advantage over many PQC schemes in terms of key size, which can lead to reduced bandwidth consumption and lower communication overhead when deployed in practical settings such as TLS. This suggests that, with further optimization, POKEx has the potential to become a competitive candidate for post-quantum key establishment.

\bibliographystyle{IEEE}
\bibliography{iso}

\begin{IEEEbiography}[{\includegraphics[width=1in,height=1.25in,clip,keepaspectratio]{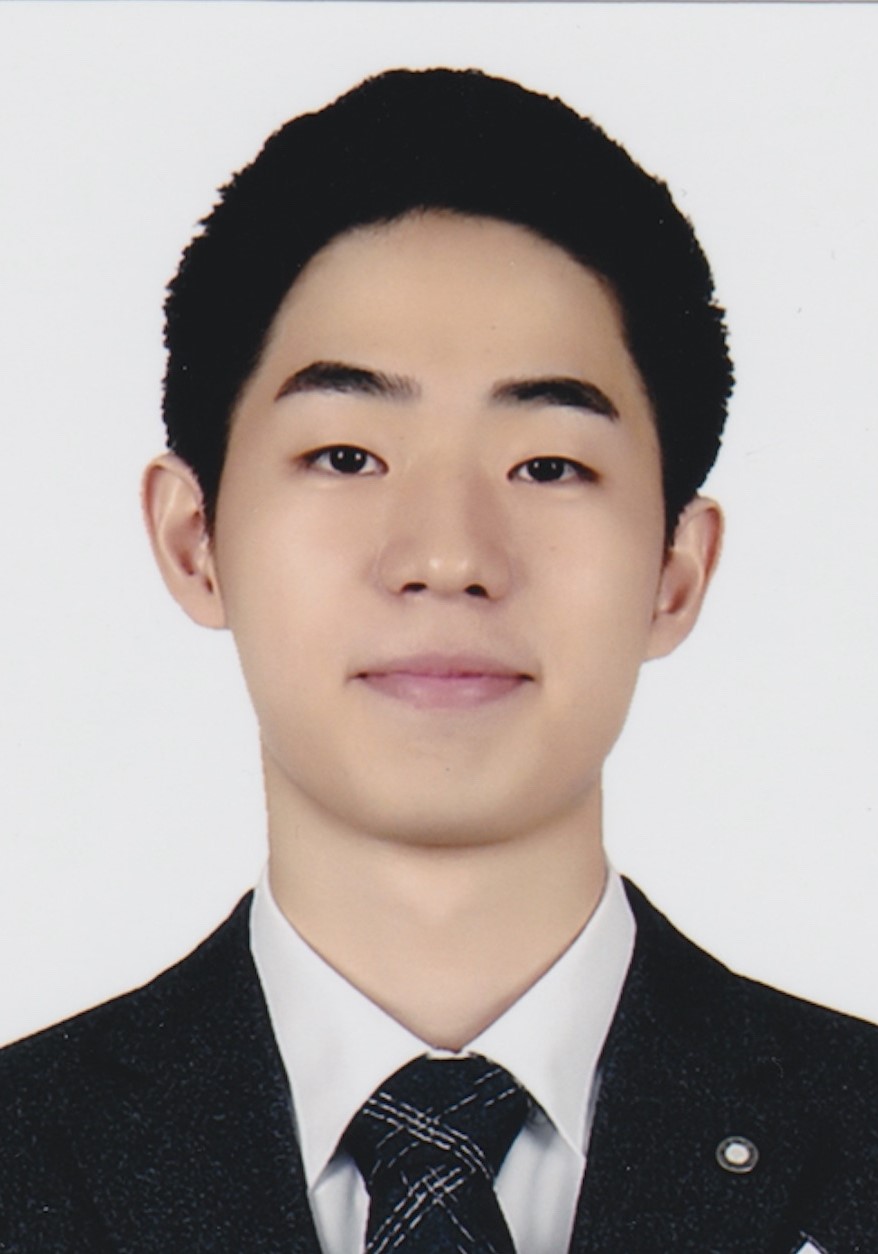}}]{Hyeonhak KIM} received the B.A. degree in cybersecurity from Korea University, in 2017. He is currently pursuing the Ph.D. degree in the School of Cybersecurity, Korea University, with a focus on post-quantum cryptography, especially isogeny-based cryptosystems, and quantum algorithms.
\end{IEEEbiography}

\begin{IEEEbiography}[{\includegraphics[width=1in,height=1.25in,clip,keepaspectratio]{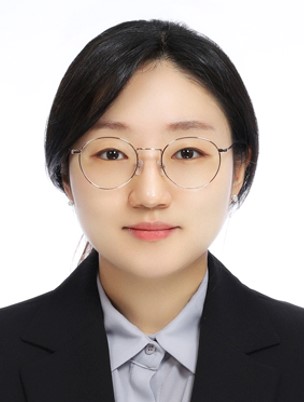}}]{Suhri KIM} received her B.A. degree in mathematics and M.A. in information security from Korea University in 2014 and 2016, respectively. She received her Ph.D. degree in the Graduate School of Information Security at Korea University in 2020. She currently holds the position of Associate Professor in the School of Mathematics, Statistics and Data Science at Sungshin Women's University. Her research focuses on post-quantum cryptography and efficient computations for isogeny-based cryptosystems.
\end{IEEEbiography}

\EOD

\end{document}